\documentclass[11pt]{llncs}
\usepackage{algorithm2e}
\usepackage{algorithmicx}
\usepackage{epsfig}
\usepackage{fancybox}
\usepackage{graphics}
\usepackage{amssymb}
\usepackage{geometry}

\newgeometry{tmargin=1in, bmargin=1in, lmargin=1in, rmargin=1in}

\begin{document}

\title{$\beta$-skeletons for a set of line segments in $R^2$
\thanks{This research is supported by the ESF EUROCORES program EUROGIGA, CRP VORONOI.}}
%\author{
%        Miroslaw Kowaluk 
%        \thanks{Institute of Informatics, University of Warsaw, Poland {\tt kowaluk@mimuw.edu.pl}}
%            \and
%        Gabriela Majewska
%        \thanks{Institute of Informatics, University of Warsaw, Poland {\tt gm248309@students.mimuw.edu.pl}}
%}
%\date{\today}
\author{
        Miros{\l}aw Kowaluk 
%\inst{1}
             \and
                     Gabriela Majewska 
%\inst{1}
\institute{
Institute of Informatics, University of Warsaw, Warsaw, Poland, \\
\texttt{kowaluk@mimuw.edu.pl} \texttt{gm248309@students.mimuw.edu.pl} 
}
}
\date{}

\maketitle

\begin{abstract}
$\beta$-skeletons are well-known neighborhood graphs for a set of points.  
We extend this notion to sets of line segments in the Euclidean plane 
and present algorithms computing such skeletons for the entire range of $\beta$ values.
The main reason of such extension is the possibility to study $\beta$-skeletons for points
moving along given line segments.
We show that relations between $\beta$-skeletons for $\beta > 1$, $1$-skeleton 
(Gabriel Graph), 
%$2$-skeleton (Relative Neighborhood Graph) 
and the Delaunay triangulation for sets of points hold also for sets of segments.
We present algorithms for computing circle-based and lune-based $\beta$-skeletons. 
We describe an algorithm that for $\beta \geq 1$ computes the $\beta$-skeleton for a set $S$ 
of $n$ segments in the Euclidean plane in $O(n^2 \alpha (n) \log n)$ time 
in the circle-based case and in $O(n^2 \lambda_4(n))$ in the lune-based one,
where the construction relies on the Delaunay triangulation for $S$, $\alpha$ is a functional inverse of Ackermann function and $\lambda_4(n)$ denotes 
the maximum possible length of a $(n,4)$ Davenport-Schinzel sequence.
When $0 < \beta < 1$, the $\beta$-skeleton can be constructed in a $O(n^3 \lambda_4(n))$ time. 
In the special case of $\beta = 1$, which is a generalization of Gabriel Graph, 
the construction can be carried out in a $O(n \log n)$ time.
\end{abstract}

\section{Introduction}
 
$\beta$-skeletons in $\mathbb{R}^2$ belong to the family of proximity graphs, 
geometric graphs in which an edge between two vertices (points) exists 
if and only if they satisfy particular geometric requirements. 
%Our definition of the $\beta$-skeleton for line segments is based on 
In this paper we use the following 
definitions of the $\beta$-skeletons for sets of points in the Euclidean space ($\beta$-skeletons are also defined for $\beta \in \{0, \infty\}$ but those cases
have no significant influence on our considerations) :

\begin{definition} 
\label{betaskeleton}
For a given set of points  $V=\{v_1,v_2, \dots ,v_n\}$ in $\mathbb{R}^2$, a distance 
function $d$  and a parameter $0 < \beta < \infty$ we define a graph 
\begin{itemize}
\item
$G_{\beta}(V)$ -- called a lune-based $\beta$-skeleton \cite{kr85} -- as follows:  
two points $v',v'' \in V$ are connected with an edge if and only if no point 
from $V \setminus \{v',v''\}$ belongs to the set 
$N(v',v'',\beta)$ (neighborhood) where:

\begin{enumerate} 
%\item 
%for $\beta=0$,  $N(v', v'',\beta)$ is equal to the segment $v'v''$;
\item 
for $0< \beta < 1$, $N(v', v'',\beta)$ is the intersection of two discs, 
each with radius $\frac{d(v',v'')}{2 \beta}$ and having the segment $v'v''$
as a chord,
%each of them has
%radius $\frac{d(v',v'')}{2 \beta}$ and whose boundaries contain both $v'$ and $v''$;
\item 
for $1 \leq \beta < \infty$, $N(v',v'',\beta)$ is the intersection of two  
discs, each with radius $\frac{\beta d(v',v'')}{2}$, whose centers are in points 
$(\frac{\beta}{2})v'+(1-\frac{\beta}{2})v''$ and 
in $(1-\frac{\beta}{2})v'+(\frac{\beta}{2})v''$, respectively;
%\item 
%for $\beta=\infty$, $N(v',v'',\beta)$ is the unbounded strip between lines 
%perpendicular to the segment $v'v''$ and containing $v'$ and $v''$.
\end{enumerate}

\item
$G^c_{\beta}(V)$ -- called a circle-based $\beta$-skeleton \cite{e02} -- as follows:  
two points $v', v''$ are connected with an edge if and only if no point 
from $V \setminus \{v', v''\}$ belongs to the set 
$N^c(v', v'',\beta)$ (neighborhood) where:
\begin{enumerate}
\item
for $0 < \beta < 1$ there is $N^c(v',v'',\beta) = N(v', v'',\beta)$, 
%is an intersection of two discs, 
%each with radius $\frac{d(v',v'')}{2 \beta}$ and having the segment $v'v''$
%as a chord.
\item
for $1 \leq \beta$ the set $N^c(v',v'',\beta)$ is a union of two discs, 
each with radius $\frac{\beta d(v',v'')}{2}$ and having the segment $v'v''$
as a chord.
\end{enumerate}
\end{itemize}

%The neighborhood $N(v', v'',\beta)$ is called a lune. 
%(The term {\em lune}, which was used in the literature up to now, is related 
%to the complement of the lens to the disc which looks like a lune - see 
%Figure \ref{fig:region}). 
\end{definition}

\begin{figure}[htbp]
\centering
\includegraphics[scale=0.3]{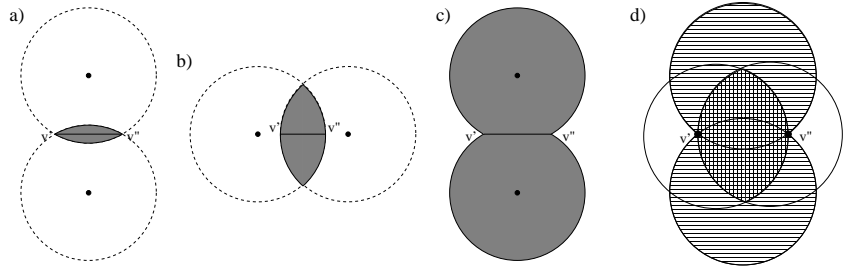}
\caption{Neighborhoods of the $\beta$-skeleton: (a) for $0 < \beta \leq 1$ , (b) the lune-based 
skeleton, (c) the circle-based skeleton for $1 < \beta < \infty$. The relation between 
neighborhoods: (d) $N(v',v'',\beta) \subseteq N^c(v',v'',\beta)$.}
\label{fig:region}
\end{figure}

%\begin{definition}
%\label{circle-based} \cite{e02}
%For a given set of points  $V=\{v_1, v_2, \dots , v_n\}$ in $\mathbb{R}^2$, a distance 
%function $d$ and a parameter $\beta \geq 0$ we define a graph 
%$G^c_{\beta}(V)$ -- called a circle-based $\beta$-skeleton -- as follows:  
%two points $v', v''$ are connected with an edge if and only if no point 
%from $V \setminus \{v', v''\}$ belongs to the set 
%$N^c(v', v'',\beta)$ where:
%\begin{enumerate}
%\item 
%for $\beta=0$,  $N^c(v', v'',\beta)$ is equal to the segment $v'v''$;
%\item
%for $0 < \beta < 1$ the set $N^c(v',v'',\beta)$ is an intersection of two discs, 
%each with radius $\frac{d(v',v'')}{2 \beta}$ and having the segment $v'v''$
%as a chord.
%\item
%for $1 \leq \beta$ the set $N^c(v',v'',\beta)$ is a union of two discs, 
%each with radius $\frac{\beta d(v',v'')}{2}$ and having the segment $v'v''$
%as a chord.
%\end{enumerate} 
%\end{definition}

Points $v',v'' \in V$ are called generators of the neighborhood $N(v',v'',\beta)$
($N^c(v',v'',\beta)$, respectively).
 The neighborhood $N(v', v'',\beta)$ is called a lune. 
It follows from the definition that $N(v',v'',1) = N^c(v',v'',1)$.
 
$\beta$-skeletons are both important and popular because of many practical applications 
which span a spectrum of areas from geographic information systems to wireless ad hoc 
networks and machine learning. 
For example, they allow us to reconstruct a shape of a two-dimensional object from 
a given set of sample points and they are also helpful in finding the minimum 
weight triangulation of a point set. 

Hurtado, Liotta and Meijer \cite{hlm03} presented an $O(n^2)$ algorithm for
the $\beta$-skeleton when $\beta< 1$.
Matula and Sokal \cite{ms84} showed that the lune-based $1$-skeleton (Gabriel Graph
$\mathit{GG}$) can be computed from the Delaunay triangulation in a linear time.
Supowit \cite{su83} described how to construct the lune-based $2$-skeleton (Relative 
Neighborhood Graph $\mathit{RNG}$) of a set of $n$ points in $O(n \log n)$ time.
Jaromczyk and Kowaluk \cite{jk87} showed how to construct the $\mathit{RNG}$ 
from the Delaunay triangulation $\mathit{DT}$ for the $L_p$ metric $(1<p<\infty)$ 
in $O(n\alpha(n))$ time.
%, where $\alpha$ is a functional inverse of Ackermann function.
This result was further improved to $O(n)$ time \cite{l94} for $\beta$-skeletons 
where $1 \leq \beta \leq 2$.
For $\beta > 1$, the  circle-based $\beta$-skeletons can be 
constructed in $O(n \log n)$ time from the  Delaunay triangulation $\mathit{DT}$ 
with a simple test to filter edges of the  $\mathit{DT}$ \cite{e02}.
On the other hand, so far the fastest algorithm for computing the lune-based 
$\beta$-skeletons for $\beta>2$ runs in $O(n^{\frac{3}{2}} \log^{\frac{1}{2}}n)$ 
time \cite{k12}.
 
%The above definitions can be used to define the $\beta$-skeleton graphs 
%for a set  $S$ of line segments $s_1, \dots , s_n$ as follows.

 Let us consider the case when we compute the $\beta$-skeleton for a set of $n$ points 
$V$ where every point $v \in V$ is allowed to move along a straight-line segment $s_v$. 
Let $S=\{s_v|v \in V\}$. 
For each pair of segments $s_{v_1}, s_{v_2}$ containing points $v_1,v_2 \in V$, 
respectively, we want to find such positions of points $v_1$ 
and $v_2$ that 
%the lune defined for this pair does not intersect any segment 
$s_v \cap N(v_1,v_2,\beta) = \emptyset$ for any $s_v \in S \setminus \{s_1,s_2\}$. 
We will attempt to solve this problem 
by defining a $\beta$-skeleton for the set of line segments $S$ as follows.

\begin{definition}
$G_{\beta}(S)$ ($G^c_{\beta}(S)$, respectively) is a graph with $n$ vertices such that there exists a bijection between the set of vertices and the set of segments $S$,
%{\bf such that each vertex corresponds to one for every segment in $S$, } 
%and two vertices determine an edge $s's''$, where $s',s'' \in S$, 
 and for $s',s'' \in S$ an edge $s's''$ exists
if there are points $v' \in s'$ and $v'' \in s''$ such that 
$(\bigcup_{s \in S \setminus \{s', s''\}} s) \cap N(v',v'',\beta) = \emptyset$
($(\bigcup_{s \in S \setminus \{s', s''\}} s) \cap N^c(v',v'',\beta) = \emptyset$, 
respectively).
\end{definition}

Note that when segments degenerate to points, we have the standard $\beta$-skeleton 
for a point set.

Geometric structures concerning a set of line segments, e.g. the Voronoi diagram 
\cite{bms84,pz13} or the straight skeleton \cite{aa96} are well-studied in the literature.

Chew and Kedem \cite{ck89} defined the Delaunay triangulation for line segments.
Their definition was generalized by Br\'{e}villiers et al. \cite{bcs07}.  

However, $\beta$-skeletons for a set of line segments were completely unexplored.
This paper makes an initial effort to fill this gap.

The paper is organized as follows.
In the next section  we present some basic facts and we prove that 
the definition of $\beta$-skeletons for a set of line segments preserves inclusions 
from the theorem of Kirkpatrick and Radke \cite{kr85} formulated for a set of points.
In Section 3 we show a general algorithm computing $\beta$-skeletons for a set of line 
segments in Euclidean plane when $0 < \beta < 1$. In Section 4 we present 
a similar algorithm for $\beta \geq 1$ in both cases of lune-based and circle-based 
$\beta$-skeletons. In Section 5 we consider an algorithm for Gabriel Graph. 
The last section contains open problems and conclusions.

\section{Preliminaries}

Let us consider a two-dimensional plane $\mathbb{R}^2$ with the Euclidean metric 
and a distance function $d$.

Let $S$ be a finite set of disjoint closed line segments in the plane.
Elements of $S$ are called sites.
A circle is tangent to a site $s$ if $s$ intersects the circle 
but not its interior.
We assume that the sites of $S$ are in general position, i.e., no three segment 
endpoints are collinear and no circle is tangent to four sites.

 The Delaunay triangulation for the set of line segments $S$ 
is defined as follows.

%Let the Delaunay triangulation $\mathit{DT}(S)$ for the set of line segments $S$ 
%be the graph (multigraph) dual to the Voronoi diagram for $S$ (see \cite{bcs07}).

\begin{definition} \cite{bcs07}
\label{segment-dt}
 The segment triangulation $P$ of $S$ is a partition of the convex hull $conv(S)$ of $S$ 
in disjoint sites, edges and faces such that:
\begin{itemize}
\item
Every face of $P$ is an open triangle whose vertices belong to three distinct sites of $S$ 
and whose open edges do not intersect $S$,
\item
No face can be added without intersecting another one,
\item
The edges of $P$ are the (possibly two-dimensional) connected components 
of $conv(S) \setminus (F \cup S)$, where $F$ is the set of faces of $P$.
\end{itemize} 

The segment triangulation $P$ such that the interior of the circumcircle of each triangle 
does not intersect $S$ is called the segment Delaunay triangulation. 

%Let $F$ be the set of triangles on the plane such that the vertices of each triangle
%belong to three distinct sites of $S$ and such that the interior of the circumcircle 
%of each triangle does not intersect $S$. 
%Then the triangles of $F$ define the faces of the segment Delaunay triangulation for $S$.
%The edges of the Delaunay triangulation are the (possibly two-dimensional) connected
%components of $conv(S) \setminus (F \cup S)$, where $conv(S)$ is the convex hull of $S$.
\end{definition}

 In this paper we will consider a planar graph (a planar multigraph, respectively) 
$\mathit{DT}(S)$ corresponding to the segment Delaunay triangulation $P$ and its relations 
with $\beta$-skeletons. This graph has a linear number of edges and is dual to the Voronoi Diagram graph for $S$.
It is also possible to study properties of plane partitions generated by $\beta$-skeletons 
for line segments. We will discuss this problem in the last section of this paper.

%\begin{remark}
%\label{lune-in-circle}
%For any $1 < \beta$ and a given set of line segments $S$ the following inclusion
%holds true $G^c_{\beta}(S) \subseteq G_{\beta}(S)$.
%\end{remark}
%\begin{proof}
%For a given parameter $1 < \beta$, neighborhoods in lune-based and circle-based 
%$\beta$-skeletons are created by circles with the same radius.
%A straightforward consequence of fact that two different circles intersect in at most 
%two points is inclusion $N(v',v'',\beta) \subseteq N^c(v',v'',\beta)$ for any pair 
%of points $v' \neq v''$ (see Figure \ref{fig:region}). 
%Hence $G^c_{\beta}(S) \subseteq G_{\beta}(S)$. 
%\end{proof}

%\begin{figure}[htbp]
%\centering
%\includegraphics[scale=0.35]{./lens-circle.eps}
%\caption{$N(v',v'',\beta) \subseteq N^c(v',v'',\beta)$.}
%\label{fig:lens-circle}
%\end{figure}  

%$MST(S)$ is a subgraph of a complete weighted graph, whose vertices correspond 
%to elements of $S$ and a weight of an edge is equal to the shortest distance 
%between corresponding segments. 

%\begin{lemma}
%\label{mst-in-dt}
%For a given set of line segments $S$ there is $MST(S) \subseteq \mathit{DT}(S)$. 
%\end{lemma}
%\begin{proof}
%Let us assume that there exist segments $s_1, s_2 \in S$ such that $s_1s_2 \in MST(S)$
%and $s_1s_2 \notin \mathit{DT}(S)$. Then the shortest path between $s_1$ and $s_2$ intersects 
%some segment $s_3 \in S$ or a triangle from $\mathit{DT}(S)$. In both cases we can replace 
%$s_1s_2$ by an edge connecting $s_1$ or $s_2$ with $s_3$ or some vertex of the triangle
%so that we obtain a tree with smaller weight than the previous one.  
%\end{proof}

 We can  consider open (closed, respectively) neighborhoods $N(v',v'',\beta)$
that lead to {\em open} ({\em closed}, respectively) {\em $\beta$-skeletons}. For example, 
the {\em Gabriel Graph} $\mathit{GG}$ \cite{gs69} is the closed $1$-skeleton and the
{\em Relative Neighborhood Graph} $\mathit{RNG}$ \cite{tou80} is the open $2$-skeleton.

Kirkpatrick and Radke \cite{kr85} showed a following important inclusions connecting 
$\beta$-skeletons for a set of points $V$ with the Delaunay triangulation 
$\mathit{DT}(V)$ of $V$ : 
%$\mathit{RNG}(V) 
$G_{\beta'}(V) \subseteq G_{\beta}(V) \subseteq \mathit{GG}(V) \subseteq \mathit{DT}(V)$, 
where $\beta' > \beta > 1$. 

%\begin{theorem} \cite{kr85}
%\label{faktinkluzja}
%Let us assume that points in $V$ are in general position.
%For $1\leq \beta \leq \beta' \leq 2$ following inclusions hold true: 
%$\mathit{MST}(V) \subseteq \mathit{RNG}(V) \subseteq G_{\beta'}(V)\subseteq G_{\beta}(V) 
%\subseteq \mathit{GG}(V) \subseteq \mathit{DT}(V)$.
%\end{theorem}

We show that definitions of the $\beta$-skeleton and the Delaunay 
triangulation for a set of line segments $S$ preserve those inclusions. 
%from Theorem \ref{faktinkluzja}.
 We define $\mathit{GG}(S)$ as a $1$-skeleton.
%($\mathit{RNG}(S)$, respectively) as a lune-based $1$-skeleton ($2$-skeleton, respectively).
   
\begin{theorem}
\label{luneinkluzja}
Let us assume that line segments in $S$ are in general position and let $G_{\beta}(S)$ 
($G^c_{\beta}(S)$, respectively) denote the lune-based (circle-based, respectively)
$\beta$-skeleton for the set $S$.
%For $1\leq \beta \leq \beta' \leq 2$
For $1\leq \beta < \beta'$ following inclusions hold true: 
%$\mathit{MST}(S) \subseteq \mathit{RNG}(S) \subseteq G_{\beta'}(S)\subseteq G_{\beta}(S) 
%\subseteq mathit{GG}(S) 
%\subseteq \mathit{DT}(S)$
$G_{\beta'}(S)\subseteq G_{\beta}(S) \subseteq \mathit{GG}(S) \subseteq \mathit{DT}(S)$
($G^c_{\beta'}(S)\subseteq G^c_{\beta}(S) \subseteq \mathit{GG}(S) \subseteq 
\mathit{DT}(S)$, respectively). 
\end{theorem}
\begin{proof} 
First we prove that $\mathit{GG}(S) \subseteq \mathit{DT}(S)$.
Let $v_1 \in s_1, v_2 \in s_2$ be such a pair of points that there exists a disc $D$ 
with diameter $v_1v_2$ containing no points belonging to segments 
from $S \setminus \{s_1, s_2\}$ inside of it. 
We transform $D$ under a homothety with respect to $v_1$ so that its image $D'$ is 
tangent to $s_2$ in the point $t$. Then we transform $D'$ under a homothety with respect 
to $t$ so that its image $D''$ is tangent to $s_1$ (see Figure \ref{fig:gg-dt}). 
The disc $D''$ lies inside of $D$, i.e., it does not intersect segments 
from $S \setminus \{s_1, s_2\}$, and it is tangent to $s_1$ and $s_2$ , so the center of $D''$ lies on the Voronoi Diagram $VD(S)$ edge. Hence, if the edge $s_1s_2$ belongs to $\mathit{GG}(S)$ then it also belongs 
to $\mathit{DT}(S)$. 

The last inclusion is based on a fact that for $1\leq \beta < \beta'$ and for any two points $v_1,v_2$ it is true that $N(v_1,v_2,\beta) \subseteq N(v_1,v_2,\beta')$ (see \cite{kr85}).

%Let $c_1, c_2$ be centers of discs determining a lune $N \in N(s_1,s_2,\beta)$.
%Let $h_v^f$ denote a homothety with respect to a point $v$ and a ratio $f$.
%Let $f = \frac{\beta'}{\beta}$, $c'_1=h_{v_1}^f(c_1)$ and $c'_2=h_{v_1}^f(c_2)$. 
%Let $c_1'$ ($c_2'$, respectively) be an image of $c_1$ ($c_2$, respectively)
%by homothety with the ratio $\frac{\beta'}{\beta}$ with respect to point $v_1$ 
%($v_2$, respectively).
%Then $c_1', c_2'$ are centers of discs determining a lune $N' \in N(s_1,s_2,\beta')$
%and $N \subseteq N'$. Hence $G_{\beta'}(S) \subseteq G_{\beta}(S)$. 

%The last inclusion can be proved (due lemma \ref{mst-in-dt}) in the same way 
%as in the paper by Kirkpatrick and Radke \cite{kr85}. 

The sequence of inclusions for circle-based $\beta$-skeletons is a straightforward 
consequence of the fact that two different circles intersect in at most two points.

\begin{figure}[htbp]
\centering
\includegraphics[scale=0.3]{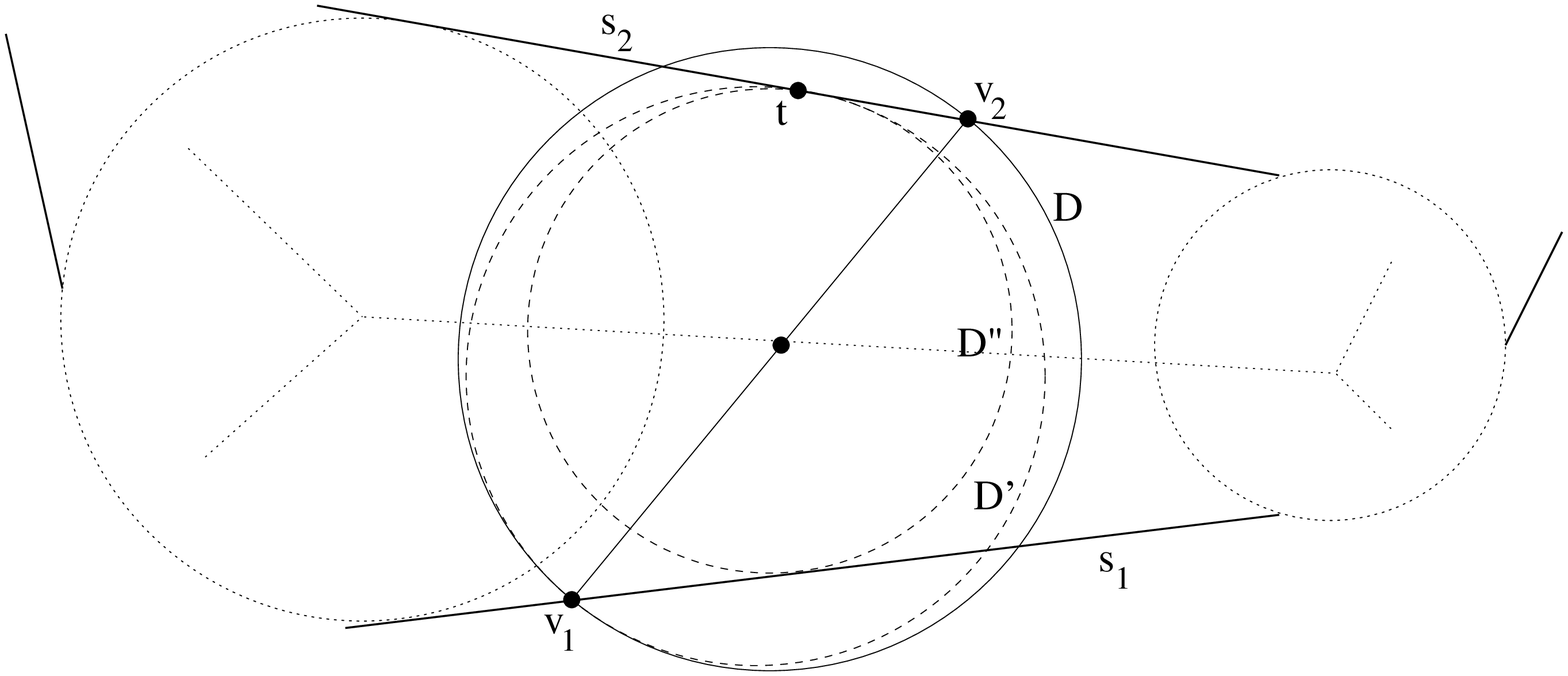}
\caption{$\mathit{GG}(S) \subseteq \mathit{DT}(S)$.}
\label{fig:gg-dt}
\end{figure}  

\end{proof}

\section{Algorithm for computing $\beta$-skeletons for $0 < \beta < 1$}

Let us consider a set $S$ of $n$ disjoint line segments in the Euclidean plane. 
First we show a few geometrical facts concerning $\beta$-skeletons 
$G_{\beta}(S)$.

%We assume that the segments are disjoint but the presented algorithm also works 
%for intersecting segments.
%For $0 < \beta < 1$, the $\beta$-skeleton is not included in $DT$.

The following remark is a straightforward consequence of the inscribed angle theorem. 

\begin{remark}
\label{lematkat}
For a given parameter $0 < \beta \leq 1$ if $v$ is a point on the boundary of  
$N(v_1,v_2, \beta)$, different from $v_1$ and $v_2$, then an angle $\angle v_1vv_2$ 
has a constant measure which depends only on $\beta$.
\end{remark}
%\begin{proof} 
%For a given pair $v_1,v_2$ the measure of such angle is constant for any $v$ since 
%inscribed angles subtended on the same arc $v_1v_2$ are equal. Now note that this angle 
%for a given pair $v_1,v_2$ depends only on the ratio of the lengths of the respective 
%two arcs on one of the discs defining the lune for $v_1,v_2$. From the definition 
%of the lune-based $\beta$-skeleton this ratio is constant and depends only on $\beta$.

%$\angle v_1vv_2$ is an inscribed angle subtended on the arc $v_1v_2$ of the disc 
%generating the lune $N(v_1,v_2,\beta)$. All such angles have the same measure. 
%For the fixed distance $d(v_1,v_2)$ a size of the angle depends on a radius 
%of the disc, i.e., it depends on $\beta$.
%\end{proof}

Let us consider a set of parametrized lines containing given segments. 
A line $P(s_i)$ contains a segment $s_i \in S$ and has a parametrization 
$q_i(t_i)=(x_1^i,y_1^i)+t_i \cdot [x_2^i-x_1^i,y_2^i-y_1^i]$, 
where $(x_1^i,y_1^i)$ and $(x_2^i,y_2^i)$ are ends of the segment $s_i$ 
and $t_i \in \mathbb{R}$. 

  Let respective points from segments $s_1$ and $s_2$ be generators of a lune and let  
%let $\delta$ be the measure of 
an inscribed angle determining a lune for a given value of $\beta$ be equal to $\delta$.
The main idea of the algorithm is as follows.
For any point $v_1 \in P(s_1)$ we compute points $v_2 \in P(s_2)$ for which 
there exists a point $v \in P(s)$, where $s \in S \setminus \{s_1, s_2\}$,
such that 
%an angle $\angle v_1vv_2$ has a measure not smaller than $\delta$ and not greater than $\pi$. 
%such that an angle $\angle v_1vv_2$ has a measure equal to the measure of inscribed
%angle determing a lune for a given value of $\beta$ and its intersection with 
%the segment $s$ is not empty. 
%We consider both positions of the inscribed angle
%(for arcs of both circles), i.e., we find points $v_2$ such that 
%$\delta \leq \angle v_1vv_2 \leq 2\pi-\delta$ (see Figure \ref{fig:hiper}).
$\delta \leq \angle v_1vv_2 \leq 2\pi-\delta$, i.e., $v \in N(v_1,v_2,\beta)$
(see Figure \ref{fig:hiper}).  
Then we analyze a union of pairs of neighborhoods generators for all $s \in S \setminus \{s_1, s_2\}$. 
If this union contains all pairs of points $(v_1,v_2)$, where $v_1 \in s_1$ and $v_2 \in s_2$,
then $(s_1,s_2) \notin G_{\beta}(S)$.

For a given $t_1 \in \mathbb{R}$ and a segment $s \in S \setminus \{s_1,s_2\}$
we shoot rays from a point $v_1=q_1(t_1) \in P(s_1)$ towards $P(s)$.
Let us assume that a given ray intersects $P(s)$ in a point 
 $v=q(t)=(x_1,y_1)+t \cdot [x_2-x_1,y_2-y_1]$ for some value of $t \in \mathbb{R}$. 
Let $w(t) = \overrightarrow{v_1v}$ be the vector between points $v_1$ and $v$. 
Then $w(t)=[A_1t+B_1t_1+C_1,A_2t+B_2t_1+C_2]$ where coefficients $A_i,B_i,C_i$ for $i=1,2$ 
%are fixed that 
depend only on endpoints coordinates of segments $s_1$ and $s$.
%The ray refracts in $v$ from the segment $s$ in a special way. 
%The sum of the angle of incidence and the angle of refraction is equal to $\delta$. 
%$= \alpha(\beta)$ for a lune generated with a given value of $\beta$.
The ray refracts in $v$ from $P(s)$ in such a way that the angle between
directions of incidence and refraction of the ray is equal to $\delta$. 
%The parametrized equation of the refracted ray is 
%$r(z,t)=v+z \cdot R^{cw}_{\delta}w(t)$ for $z \geq 0$ 
%(or $r(z,t)=v+z \cdot R^{ccw}_{\delta}w(t)$ for $z \geq 0$) 
%where $R^{cw}_{\delta}$ ($R^{ccw}_{\delta}$, respectively) denote a rotation matrix 
%for a clockwise (counter-clockwise, respectively) angle $\delta$. 
The parametrized equation of the refracted ray is 
$r(z,t)=v+z \cdot R_{\delta}w(t)$ for $z \geq 0$ 
(or $r(z,t)=v+z \cdot R'_{\delta}w(t)$ for $z \geq 0$, respectively) 
where $R_{\delta}$ ($R'_{\delta}$, respectively) denotes a rotation matrix 
for a clockwise (counter-clockwise, respectively) angle $\delta$. 
If refracted ray $r(z,t)$ intersects line $P(s_2)$ in a point $q_2(t_2)=r(z,t)$ 
(it is not always possible - see Figure \ref{fig:hiper}) then we compare the 
$x$-coordinates of $q_2(t_2)$ and $r(z,t)$. 
As a result we obtain a function containing only parameters $t_1$ and $t_2$:  
$z=\frac{J \cdot t_2+K \cdot t_1+L}{D \cdot t+E \cdot t_1+F}$, 
where coefficients $J = -(x_2-x_1) ,K = x_2^2-x_1^2 ,L = x_1^2-x_1 ,
D = A_1 \cos \delta + A_2 \sin \delta,E = B_1 \cos \delta + B_2 \sin \delta,
F = C_1 \cos \delta + C_2 \sin \delta$ are fixed.
Since $y$-coordinates of $q_2(t_2)$ and $r(z,t)$ are also equal we obtain 
$t_2(t)=\frac{M \cdot t^2+p_1(t_1) \cdot t+p_2(t_1)}{N \cdot t+p_3(t_1)}$, 
where $p_1,p_2$ and $p_3$ are (at most quadratic) polynomials of variable 
$t_1$ and $M,N$ are fixed (the exact description of those polynomials and variables
is much more complex than the description of the coefficients in the previous step and it is omitted here 
- see them in Appendix).
 
%Note that if we consider the clockwise and the counter-clockwise refraction separately
%then for a given point $q_1(t_1)$ and a given segment $s$, the graph of function $t_2$ 
%with respect to $t$ consists of parts of a hyperbola.  
%We analyze a graph of correlations between variables $t$ and $t_2$ 
%(i.e., a set of pairs $(t,t_2(t))$ for fixed $t_1$) for both kinds 
%of refractions - see Figure \ref{fig:hiper}.

 Let $l^\delta_{t_1}(t)$ denote a value of the parameter $t_2$ of the intersection point
of the line $P(s_2)$ and the line containing the ray that starts in $q_1(t_1)$ and refracts
in $q(t)$ creating an angle $\delta$. Let $k^\delta_{t_1} = l^\delta_{t_1}\huge |_I$, where $I$ is 
a set of values of $t$ such that the ray refracted in $q(t)$ intersects $P(s_2)$. 
The function $l^\delta_{t_1}$ is a hyperbola and the function $k^\delta_{t_1}$ is a part of it
(see Figure \ref{fig:hiper}).

%For a given $t_1$ we want to find correlation between $t_2$ 
%and $t$ for $t \in \mathbb{R}$. 

\begin{figure}[htbp]
\centering
\includegraphics[scale=0.34]{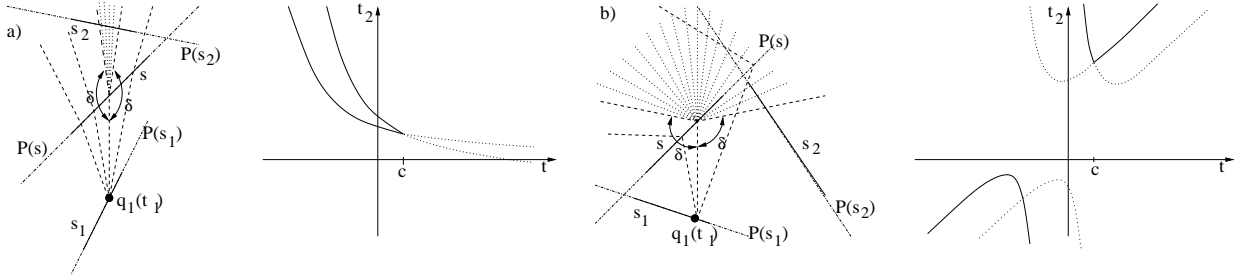}
\caption{Examples of correlation between parameters $t$ and $t_2$ (for a fixed $t_1$) 
for a presented composition of segments and (a) a refraction angle near $\pi$ 
(dotted lines show refracted rays that are analyzed) and (b) near $\frac{\pi}{2}$. 
The value $c$ corresponds to the intersection point of lines $P(s)$ and $P(s_2)$.
 Dotted curves show a case when a line containing a refracted ray intersects 
$P(s_2)$ but the ray itself does not. }
\label{fig:hiper}
\end{figure}

 Note that for a given angle $\delta$ ($2\pi - \delta$, respectively) extreme points
of the function  $k^\delta_{t_1}$ ($k^{2\pi - \delta}_{t_1}$, respectively) do not have to 
belong to the set $\{0,1\}$. We can find them by computing a derivative
$\frac{dt_2}{dt}=\frac{MN \cdot t^2+2Mp_3(t_1) \cdot t+p_1(t_1)p_3(t_1)-Np_2(t_1)}
{(N \cdot t+p_3(t_1))^2}$.

Then we can compute the corresponding values of the parameter $t_2$. 
This way we obtain the pair $(t_1, t_2)$ such that the segment $q_1(t_1)q_2(t_2)$
is a chord of a circle that is tangent to the analyzed segment $s$ in $q(t)$ and
$\angle q_1(t_1)q(t)q_2(t_2) = \delta$ ($\angle q_1(t_1)q(t)q_2(t_2) = 2\pi - \delta$,
respectively).  

%Moreover, for some values of $t$ the refracted ray $r(z,t)$ does not intersect line 
%$P(s_2)$ (both in the case when we refract the ray clockwise and counter-clockwise). 

%When we consider both clockwise and counter-clockwise refraction the points we get will lie 
%on two hyperbolas and the asymptotes of those hyperbolas will determine all such values of $t$.
%Now, to check if segment $s$ eliminates all lunes $N_2^s(s_1, s_2,\beta)$ 

%Let $T(t_1,s,s_2)$ be a set of all $t_2$ such that, for given $t_1$, points $q_1(t_1)$ 
%and $q_2(t_2)$ generate a lune intersected by segment $s$.  

 Let $T(t_1,s,s_2) = \bigcup_{\gamma \in [\delta, 2\pi - \delta], t \in [0, 1]} k^\gamma_{t_1}(t)$,
i.e., this is a set of all $t_2$ such that points $q_1(t_1)$ and $q_2(t_2)$ generate 
a lune intersected by the analyzed segment $s$.
Let $F(s_1,s,s_2)= \bigcup_{t_1 \in R, x \in T(t_1,s,s_2)}(t_1,x)$ be a set of pairs 
of parameters $(t_1,t_2)$ such that the segment $s$ intersects a lune generated by points 
$q_1(t_1)$ and $q_2(t_2)$.
The set $F(s_1,s,s_2)$ is an area limited by $O(1)$ algebraic curves 
of degree at most $3$. 
%(i.e., hyperbolas for $t=0$ and $t=1$ and curves corresponding 
%to extreme points of analyzed hyperbolas and their intersection points). 
 The curves match the values of the parameter $t_2$ corresponding 
to extreme points of $k^\delta_{t_1}$ ($k^{2\pi - \delta}_{t_1}$, respectively). In particular 
there are hyperbolas for angles $\delta$ and $2\pi - \delta$ correlated with the rays refracted 
in the ends of the segment $s$ (for parameters $t = 0$ and $t = 1$) - see 
Figure \ref{fig:polygon}.

\begin{figure}[htbp]
\centering
\includegraphics[scale=0.33]{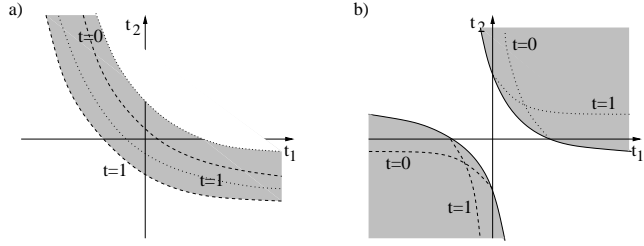}
\caption{ Examples of sets $F(s_1,s,s_2)$ for $\beta$ near (a) $0$ and (b) $1$ 
(the shape of $F(s_1,s,s_2)$ also depends on the position of the segment $s$ with respect 
to $s_1$ and $s_2$). Dotted (dashed, respectively) curves limit the area corresponding 
to rays refracted through the segment $s$ and creating the angle $\delta$ ($2\pi - \delta$, 
respectively).} 
\label{fig:polygon}
\end{figure}

\begin{lemma}
The edge $s_1,s_2$ belongs to the $\beta$-skeleton $G_{\beta}(S)$ if and only if \\
$[0,1] \times [0,1] \setminus \bigcup_{s \in S \setminus \{s_1,s_2\}} F(s_1,s,s_2) \neq \emptyset$. 
\end{lemma}
\begin{proof}
If $[0,1] \times [0,1] \setminus \bigcup_{s \in S \setminus \{s_1,s_2\}} F(s_1,s,s_2) \neq \emptyset$
then there exists a pair of parameters $(t_1,t_2) \in [0,1] \times [0,1]$ such that 
a lune generated by points $q_1(t_1) \in s_1$ and $q_2(t_2) \in s_2$ is not intersected 
by any segment $s \in S \setminus \{s_1,s_2\}$, i.e., $(s_1,s_2) \in G_{\beta}(S)$.
The opposite implication can be proved in the same way.
\end{proof}

\begin{theorem}
\label{lambda4}
For $0 < \beta < 1$ the $\beta$-skeleton $G_{\beta}(S)$ can be found 
in $O(n^3 \lambda_4(n))$ time.
%, where $\lambda_4(n)$ denotes the maximum possible length 
%of a $(n,4)$ Davenport-Schinzel sequence. 
\end{theorem}
\begin{proof}
We analyze $O(n^2)$ pairs of line segments. For each pair of segments $s_1,s_2$ 
we compute $\bigcup_{s \in S \setminus \{s_1,s_2\}} F(s_1,s,s_2)$. 
For each $s \in S \setminus \{s_1,s_2\}$ we find a set of pairs of parameters
$t_1,t_2$ such that $N(q_1(t_1),q_2(t_2),\beta) \cap s \neq \emptyset$.
The arrangement of $n-2$ curves in total can be found in $O(n \lambda_4(n))$ time
\cite{gr04}. 
Then the difference 
$[0,1] \times [0,1] \setminus \bigcup_{s \in S \setminus \{s_1,s_2\}} F(s_1,s,s_2)$
can be found in $O(n \lambda_4(n))$ time (see Figure \ref{fig:omega4} in Appendix). 
Therefore we can verify which edges belong to $G_{\beta}(S)$ in $O(n^3 \lambda_4(n))$ time. 
\end{proof}

%\begin{figure}[htbp]
%\centering
%\includegraphics[scale=0.25]{./omega4.eps}
%\caption{
%An example of the set $S$ where for segments $s_1, s_2$ the difference 
%$[0,1] \times [0,1] \setminus \bigcup_{s \in S \setminus \{s_1,s_2\}} F(s_1,s,s_2)$
%contains $\Omega(n^2)$ connected components (for very small $\beta$). 
%}
%\label{fig:omega4}
%\end{figure}

%Let $R(s_1,s_2,s)$ (or $R(s_2,s_1,s)$, respectively) be the graph of the function 
%$F(t_1)=T(t_1,s,s_2)$ (or $F(t_2)=T(t_2,s,s_1)$).
%Segment $s$ eliminates all lunes $N_2^s(s_1, s_2,\beta)$ if and only if set 
%$R(s_1,s_2,s) \cap R(s_2,s_1,s)$ covers the whole square $[0,1]^2$.
%We repeat this procedure for all segments $s \in S \setminus \{s_1,s_2 \}$ and check if the sum 
%of sets we got this way covers the whole square $[0,1]^2$. If not then there exists an empty lune 
%$N_p^s(s_1, s_2,\beta)$ and edge $s_1s_2$ belongs to the $\beta$-skeleton for $S$.

\section{Finding $\beta$-skeletons for $1 \leq \beta$}

Let us first consider the circle-based $\beta$-skeletons. According to Theorem \ref{luneinkluzja} for $1 \leq \beta$ there are only $O(n)$
edges which can belong to the $\beta$-skeleton for a given set of line segments.
  We will use this property to compute $\beta$-skeletons faster than in the previous 
section.  
%Moreover, the number of connected components of the set 
%$[0,1] \times [0,1] \setminus \bigcup_{s \in S \setminus \{s_1,s_2\}} F(s_1,s,s_2)$
%is smaller than for $\beta < 1$. Therefore we can compute $\beta$-skeletons faster
%than in the previous case. 

\begin{lemma}
\label{components}
For $1 \leq \beta$ and the set $S$ of $n$ line segments the number of connected components 
of the set $[0,1] \times [0,1] \setminus \bigcup_{s \in S \setminus \{s_1,s_2\}} F(s_1,s,s_2)$
is O(n) for any pair $s_1,s_2 \in S$.
\end{lemma}
\begin{proof}
 According to Theorem by Kirkpatrick and Radke \cite{kr85} for $1 \leq \beta < \beta'$ 
the following inclusion holds $G_{\beta'}(v) \subseteq G_{\beta}(V)$.
Therefore any neighborhood for $\beta'$ is included in some neighborhood for $\beta$ with 
the same pair of generators.
On the other hand, for a given parameter $\beta$ and a given connected component 
of the set $[0,1] \times [0,1] \setminus \bigcup_{s \in S \setminus \{s_1,s_2\}} F(s_1,s,s_2)$
there exists a sufficiently big $\beta'$ such that for $\beta'$ the component contains only one point
(we increase an arbitrary neighborhood corresponding to the connected component for 
a given $\beta$).
Hence, the number of one point components (for all values of $\beta$) estimates
the number of connected components for a given $\beta$.
But in this case at least one disc forming the neighborhood is tangent to two
segments different than $s_1$ and $s_2$ or at least one generator of the neighborhood
is at the end of $s_1$ or $s_2$.
In the first case the two segments tangent to the disc and segments $s_1, s_2$ are the
the closest ones to the center of the disc. Therefore the complexity
of the set of such components does not exceed the complexity of the $4$-order
Voronoi diagram for $S$, i.e., it is $O(n)$ \cite{pz13}.
In the second case there is a constant number of additional components.   

%Let us consider neighborhoods for circle-based $\beta$-skeletons.
%The neighborhood is {\bf an union of} two discs. The position of one of them uniquely defines 
%the position of the second one. Let us consider neighborhoods generated by edges between
%$s_1$ and $s_2$ corresponding to points on the boundary of a given connected component.
%The {\bf neighborhood} is defined by at most four segments ({\bf the corresponding discs 
%intersect $s_1, s_2$ in points generating the neighborhood and are tangent to at most two
%other segments}).
%Those segments are the closest ones to the center of the disc. Therefore the complexity
%of the set of connected components does not exceed the complexity of the $4$-order
%Voronoi diagram for $S$, i.e., it is $O(n)$ \cite{pz13}.    
\end{proof}

\begin{lemma}
\label{two-side}
For any $t_1 \in \mathbb{R}$ and $s_1, s_2 \in S$ there is at most one connected component of the set
$[0,1] \times [0,1] \setminus \bigcup_{s \in S \setminus \{s_1,s_2\}} F(s_1,s,s_2)$
that contains points with the same $t_1$ coordinate. 
\end{lemma}
\begin{proof}
 Let the inscribed angle corresponding to $N^c(s_1,s_2,\beta)$ be equal to $\delta$.
Let $a = q_1(t_1)$ and $b \in P(s_2)$ ($b' \in P(s_2)$, respectively) be points 
such that the angle between $ab$ ($ab'$, respectively) and $P(s_2)$ is equal to $\delta$
(for $\delta = \frac{\pi}{2}$ we have $b = b'$), see Figure \ref{fig:two-side}.
 Boundaries of all neighborhoods $N^c(s_1,s_2,\beta)$ 
%having $a$ on its boundary intersect in $b$ or $b'$.
generated by $a$ and a point in $s_2$ contain either $b$ or $b'$. 
There exists the leftmost (rightmost, respectively) position (might be in infinity) 
of the second neighborhood generator with respect to the direction of $t_2$. Between 
those positions no neighborhood intersects segments from $S \setminus \{s_1, s_2\}$. 
 Hence, points corresponding to positions of such generators 
%of such neighborhoods 
belong to the same connected component of 
$[0,1] \times [0,1] \setminus \bigcup_{s \in S \setminus \{s_1,s_2\}} F(s_1,s,s_2)$. 
\end{proof}

\begin{figure}[htbp]
\centering
\includegraphics[scale=0.4]{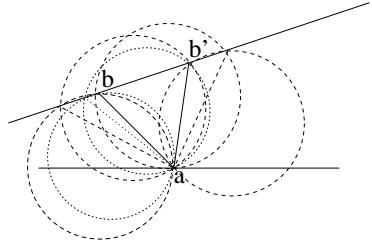}
\caption{Neighborhoods that have one common generator.}
\label{fig:two-side}
\end{figure}

The algorithm for computing circle-based $\beta$-skeletons for $\beta \geq 1$ is almost
the same as the algorithm for $\beta < 1$. 
%We use the divide and conquer method.
%However, there exist two sets of pairs $(t_1,t_2)$ 
%such that some analyzed segment from $S$ intersects a neighborhood 
%$N^c(q_1(t_1),q_2(t_2),\beta)$ (see Figure \ref{fig:circle}). 
%The sets are separated by a line $t_1=c$, for some $c \in \mathbb{R}$, and the first 
%of them is limited from below and the second one is limited from above.
%Therefore we can easy compute intersection of complements of those sets. 

\begin{theorem}
For $\beta \geq 1$ the circle-based $\beta$-skeleton $G^c_{\beta}(S)$ can be found 
in $O(n^2 \alpha(n) \log n)$ time. 
%where $\alpha(n)$ is the inverse Ackermann function.
\end{theorem} 
\begin{proof}
 Due to Theorem \ref{luneinkluzja} we have to analyze $O(n)$ edges of $\mathit{DT}(S)$.
%According to Lemma \ref{two-side} for each edge we use Hershberger's algorithm \cite{h89} 
%to compute the intersection of the complements of sets containing pairs of points generating not 
%empty neighborhoods. We find the lower envelope of curves intersecting the upper edge of the square
%$[0,1] \times [0,1]$ and the upper envelope of curves intersecting the lower edge of the square. 
%Then we intersect sets limited by those envelopes. 
 For $\beta \geq 1$ and for the given segments $s_1, s_2 \in S$ each set $F(s_1,s,s_2)$
can be divided in two sets with respect to the variable $t_1$. For each $t_1$ the first 
set contains part of $F(s_1,s,s_2)$ that is unbound from above with respect to $t_2$ 
and the second one contains part of $F(s_1,s,s_2)$ unbound from below.
The part that contains pairs $(t_1,t_2)$ such that the set of values of $t_2$ is $\mathbb{R}$ can be divided arbitrarily.
We use Hershberger's algorithm \cite{h89} to compute unions of sets for 
$s \in S \setminus \{s_1,s_2\}$ in each group separately. 
Then, according to Lemma \ref{two-side} we find an intersection of complements of computed
unions.  
It needs $O(n\alpha(n) \log n)$ time. Hence, the total time complexity of the algorithm 
is $O(n^2\alpha(n) \log n)$.  
\end{proof}

\begin{figure}[htbp]
\centering
\includegraphics[scale=0.33]{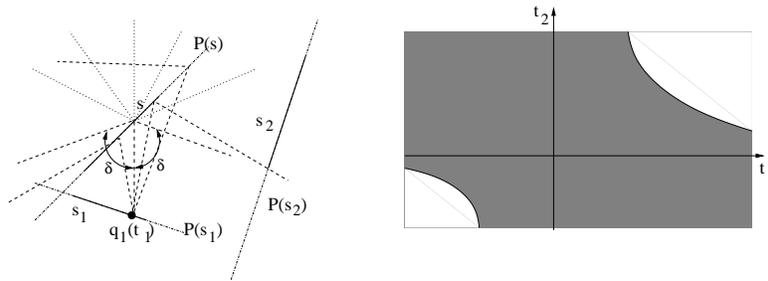}
\caption{An example of (a) refracted rays and (b) correlations between variables $t_1$ 
and $t_2$ for circle-based $\beta$-skeletons, where $\beta \geq 1$.
(the shape of $F(s_1,s,s_2)$ depends on the position of the segment $s$ with respect 
to $s_1$ and $s_2$)}
\label{fig:circle}
\end{figure}

Let us consider the lune-based $\beta$-skeletons now. Unfortunately, Lemma \ref{two-side}
does not hold in this case (see Figure \ref{fig:paski} in Appendix).

%\begin{figure}[htbp]
%\centering
%\includegraphics[scale=0.3]{./paski.eps}
%\caption{An example of the segment position (for $\beta$ near to infinity) that 
%the same point generates lunes which correspond to points in different connect components.}
%\label{fig:paski}
%\end{figure}   

According to Theorem \ref{luneinkluzja}, in this case we have to consider only $O(n)$
pairs of line segments in $S$ (the pairs corresponding to edges of $\mathit{DT}(S)$).
 We will analyze pairs of points belonging to given segments $s_1,s_2 \in S$
which generate discs such that each of them is intersected by any segment 
$s \in S \setminus \{s_1,s_2\}$.
We will consider $\beta$-skeletons for $\beta > 1$ (a $1$-skeleton is the same 
in the circle-based and lune-based case). 
Let $q_1(t_1) \in s_1$ and $q_2(t_2) \in s_2$ be generators of a lune 
$N(q_1(t_1),q_2(t_2),\beta)$ and let $C_1(q_1(t_1),q_2(t_2),\beta)$ be a circle creating 
a part of its boundary containing point $q_1(t_1)$. 

%For $\beta = 1$ the angle refracted in $q(t)$ is equal to $\frac{\pi}{2}$.
%For $\beta > 1$ the angle leaves right but we analyze a homothetic 
%(with respect to a point $q_1(t_1)$ and then to a point $q_2(t_2)$)  
%images of segments which could intersect 
%lunes generated by two given segments (see Figure \ref{fig:tales}).
We will shoot a ray from a lune generator and we will compute a possible position
of the second generator when the refraction point belongs to the lune.
Let an angle between a shot ray and a refracted ray be equal to $\frac{\pi}{2}$
and let $q(t) \in s \cap C_1(q_1(t_1),q_2(t_2),\beta)$.
Unfortunately, the ray shot from $q_1(t_1)$ and refracted in $q(t)$
does not intersect the segment $s_2$ in $q_2(t_2)$.
However, we can define a segment $s'$ such that the ray shot from $q_1(t_1)$ refracts
in $q(t)$ if and only if the same ray refracted in a point of $s'$ passes through $q_2(t_2)$
(see Figure \ref{fig:tales}).   
%Points $q_1(t_1), q_2(t_2)$ are on the diameter of the disc $D$ and a point $q(t)$ 
%lies on a circle of this disc.  
  
%First, we have to prove that the following lemma is true:
%\begin{lemma}
%For a given $\beta >1$ let $v$ be a point from the boundary of a lune $N_d(v_1,v_2, \beta)$, 
%different from $v_1,v_2$. The line passing through point $v_2$ that is also parallel 
%to the bisector of the segment $v_1v$ intersects this segment in the point $w$. 
%Now the ratio $\frac{d_2(v_1,w)}{d_2(v_1,v)}$ is constant and depends only on $\beta$.
%\end{lemma}

%{\bf Na rysunku dodac $c$ !!!.}

\begin{figure}[htbp]
\centering
\includegraphics[scale=0.3]{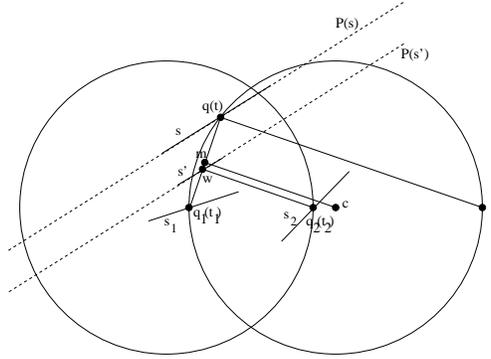}
\caption{ The auxiliary segment $s'$ and rays refracted in $q(t)$ and $w$. 
%The segment $s'$ is the homothetic image of the segment $s$ with the scale $\frac{1}{\beta}$.
}
\label{fig:tales}
\end{figure}

\begin{lemma}
 Assume that $\beta \geq 1$, $q_1(t_1) \in P(s_1)$ and $q_2(t_2) \in P(s_2)$,
where $s_1,s_2 \in S$. Let a point $q(t) \in P(s)$, where $s \in S \setminus \{s_1,s_2\}$, 
%lie on a circle containing a boundary of a lune $N(q_1(t_1),q_2(t_2),\beta)$ opposite to $q_2(t_2)$.
%lie on $N(q_1(t_1),q_2(t_2),\beta) \cap C_1(q_1(t_1),q_2(t_2),\beta)$. 
belong to $C_1(q_1(t_1),q_2(t_2),\beta)$. 
Let $l$ be a line perpendicular to the segment $(q_1(t_1),q(t))$, passing
through $q_2(t_2)$ and crossing $(q_1(t_1),q(t))$ in a point $w$. Then 
$\frac{d(q_1(t_1),w)}{d(q_1(t_1),q(t))} = \frac{1}{\beta}$.
    
%the ratio $\frac{d(q_1(t_1),w)}{d(q_1(t_1),q(t))}$ is constant and equal to $\frac{1}{\beta}$. } 
\end{lemma}
\begin{proof} 
 Let $x$ be an opposite to $q_1(t_1)$ end of the diameter of $C_1(q_1(t_1),q_2(t_2),\beta)$.    
%a disc $D$ which defines a lune $N(q_1(t_1),q_2(t_2),\beta)$ and contains $q_1(t_1)$ on its boundary.
Then $d(q_1(t_1),x) = 2d(q_1(t_1),c)$, where $c$ is the center of $C_1(q_1(t_1),q_2(t_2),\beta)$.
From the definition of the $\beta$-skeleton follows that 
$\frac{d(q_1(t_1),q_2(t_2))}{d(q_1(t_1),x)}=\frac{d(q_1(t_1),q_2(t_2))}{2d(q_1(t_1),c)} \cdot
\frac{2d(q_1(t_1),q_2(t_2))}{2\beta d(q_1(t_1),q_2(t_2))}=\frac{1}{\beta}$.
%The bisector of $(q_1(t_1),q(t))$ contains the point $c$.
%The line parallel to $l$ passing through $c$ divides the segment 
%$(q_1(t_1),q(t))$ in half.
According to Thales' theorem $\frac{d(q_1(t_1),w)}{d(q_1(t_1),q(t))}=
\frac{d(q_1(t_1),q_2(t_2))}{d(q_1(t_1),x)}=\frac{1}{\beta}$ (see Figure \ref{fig:tales}). 

\end{proof}

%\begin{corollary}
%Let, for a given $\beta \geq 1$, $q_1(t_1) \in P(s_1)$ and $q_2(t_2) \in P(s_2)$,
%where $s_1,s_2 \in S$.  
%Let a point $w$ belongs to the segment $q_1(t_1)q(t)$ and 
%$\frac{d(q_1(t_1),w)}{d(q_1(t_1),q(t))}=\frac{1}{\beta}$.
%Then, a point $q(t) \in P(s)$, where $s \in S \setminus \{s_1,s_2\}$, 
%lies on a circle containing a boundary of a lune $N(q_1(t_1),q_2(t_2),\beta)$ opposite 
%to $q_2(t_2)$ if and only if $\angle q_1(t_1)wq_2(t_2) = \frac{\pi}{2}$. 
%\end{corollary}

%\begin{corollary}
%For Gabriel graphs the ratio is equal to $1$ and for relative neighborhood graphs 
%it is equal to $\frac{1}{2}$.
%\end{corollary}

 The algorithm computing a lune-based $\beta$-skeleton for $\beta \geq 1$ is similar 
to the previous one. Let $P(s') = h^{\frac{1}{\beta}}_{q_1(t_1)}(P(s))$, where
$h^{\frac{1}{\beta}}_{q_1(t_1)}$ is a homothety with respect to a point $q_1(t_1)$ 
and a ratio $\frac{1}{\beta}$.
Like in the case of circle-based $\beta$-skeletons we compute pairs of parameters $t_1, t_2$ 
such that the ray shot from $q_1(t_1)$ refracts in a point of $s'$ and intersects the segment
$s_2$ in $q_2(t_2)$, i.e., an analyzed segment $s$ intersects a disc limited by the circle 
$C_1(q_1(t_1),q_2(t_2),\beta)$. 

%We consider a line $P'(s)$ which is a homothetic image of the line $P(s)$ with respect 
%to a point $q_1(t_1)$ and a ratio $\frac{1}{\beta}$. 
%The ray shot from $q_1(t_1)$ refracts off $P'(s)$ at a right angle.
However, in the case of lune-based $\beta$-skeletons we analyze only one hyperbola 
(functions for clockwise and counterclockwise refractions are the same). 
Moreover, sets $F(s_1,s,s_2)$ and $F(s_2,s,s_1)$ are different. 
 They contain pairs of parameters $t_1, t_2$ corresponding to points generating discs 
such that each of them separately is intersected by the segment $s$.
Therefore, we have to intersect those sets to obtain a set of pairs of parameters
corresponding to points generating lunes intersected by $s$ (see Figure \ref{fig:right}). 

%In a such way we obtain unbounded sets, limited by $O(1)$ algebraic curves of degree 
%at most 3 which also have vertical and horizontal asymptotes .

\begin{figure}[htbp]
\centering
\includegraphics[scale=0.3]{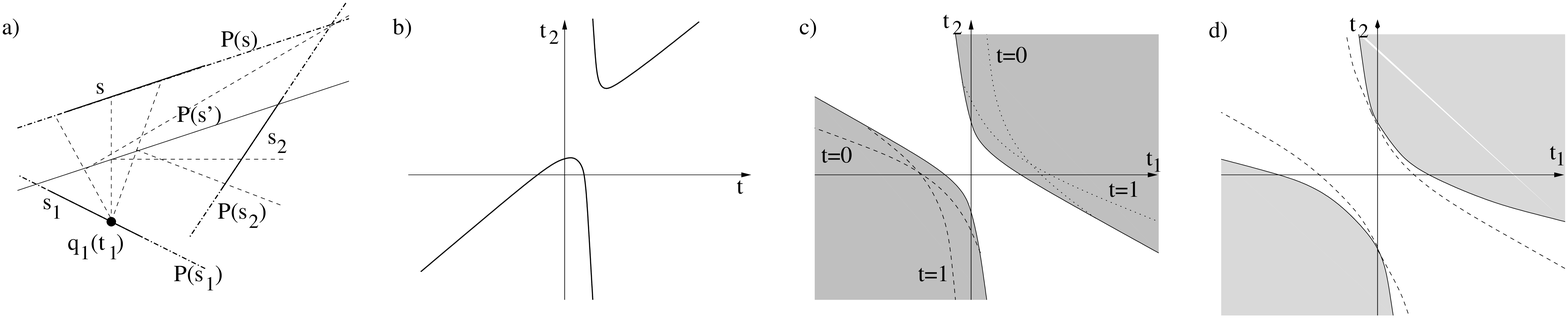}
\caption{An example of (a) a composition of three segments $s_1, s_2, s$, (b) correlations 
between variables $t$ and $t_2$ (parametrizing $s$ and $s_2$, respectively), 
(c) the set $F(s_1,s,s_2)$ and (d) the intersection $F(s_1,s,s_2) \cap F(s_2,s,s_1)$,
where $\beta > 1$. }
%(the shape of $F(s_1,s,s_2)$ depends on the position of the segment $s$ with respect 
%to $s_1$ and $s_2$)}
\label{fig:right}
\end{figure}

\begin{theorem}
For $\beta \geq 1$ the lune-based $\beta$-skeleton $G_{\beta}(S)$ can be found 
in $O(n^2 \lambda_4(n))$ time.
%, where $\alpha(n)$ is the inverse Ackermann function.
\end{theorem} 
\begin{proof}
$\beta$-skeletons for $\beta \geq 1$ satisfy the inclusions from Theorem \ref{luneinkluzja}.
Hence, the number of tested edges is linear. For each such pair of segments $s_1,s_2$ 
we compute the corresponding sets of pairs of points generating lunes that do not intersect 
segments from $S \setminus \{s_1,s_2\}$. 
Similarly as in Theorem \ref{lambda4} we can do it in $O(n \lambda_4(n))$ time.
Therefore, the total time complexity of the algorithm (after analysis of $O(n)$ pairs 
of segments) is $O(n^2 \lambda_4(n))$.   
\end{proof}

\section{Computing Gabriel Graph for segments}

%For a given pair of segments $s_1,s_2 \in S$ and for every segment 
%$s \in S \setminus \{s_1,s_2\}$ we consider points from the intersection 
%of the $2$-Voronoi region $2-VR(s_1,s_2)$ for $s_1, s_2$ and the $3$-Voronoi region 
%$3-VR(s_1,s_2,s)$ for $s_1,s_2,s$. 

 In the previous sections we constructed sets of all pairs of points generating 
neighborhoods that do not intersect segments other than the segments containing generators. 
Now we want to find only $O(n)$ pairs of generators (one pair for each edge 
of a $\beta$-skeleton) that define the graph.   
Let $2-VR(s_1,s_2)$ denote a region of the $2$-order Voronoi diagram for the set $S$
corresponding to $s_1, s_2$ and $3-VR(s_1,s_2,s)$ denote a region of the $3$-order Voronoi 
diagram for the set $S$ corresponding to $s_1,s_2,s$.
If an edge $s_1s_2$, where $s_1, s_2 \in S$ belongs to the Gabriel Graph then there 
exists a disc $D(p,r)$ centered in $p$, which does not contain points from 
$S \setminus \{s_1,s_2\}$ and its diameter is $v_1v_2$, where $v_1 \in s_1$, 
$v_2 \in s_2$ and $2r = d(v_1,v_2)$. 
The disc center $p$ belongs to the set $(2-VR(s_1,s_2)) \cap (3-VR(s_1,s_2,s))$ 
for some $s \in S \setminus\{s_1,s_2\}$.

 First, for segments $s_1,s_2 \in S$ we define a set of all middle points 
of segments with one endpoint on $s_1$ and one on $s_2$.
%centers of discs, for which there exists a diameter connecting $s_1$ and $s_2$. 
This set is a quadrilateral $Q(s_1,s_2)$ (or a segment if $s_1 \parallel s_2$) with vertices 
in points $\frac{(x_i^1,y_i^1)+(x_j^2,y_j^2)}{2}$, where $(x_i^k,y_i^k)$ 
for $i=1,2$ are endpoints of the segments $s_k$ for $k=1,2$ (boundaries of the set 
 are determined by the images of $s_1$ and $s_2$ under four homotheties with respect to the ends of those segments 
and a ratio $\frac{1}{2}$).

 Let us analyze a position of a middle point of a segment $l$ whose ends slide 
along the segments $s_1, s_2 \in S$. Let the length of $l$ be $2r$.  
We rotate the plane so that segment $s_1$ lies in the negative part of $x$-axis 
and the point of intersection of lines containing segments $s_1$ and $s_2$ 
(if there exists) is $(0,0)$.
 Let the segment $s_2$ lie on the line parametrized by
$u \cdot [x_1,y_1]$ for $0 \geq x_1, 0 \leq y_1, 0 \leq u$. Then the middle point of $l$ 
%with endpoint on $s_2$ in $u \cdot [x_1,y_1]$ 
is $(x,y)$, where $x=-|\sqrt{r^2-(\frac{uy_1}{2})^2}|+u \cdot x_1,y=\frac{uy_1}{2}$. 

%The other endpoint of the ladder is point 
%$(0,-2|\sqrt{r^2-(\frac{sy_1}{2})^2}|+s \cdot x_1)$ and we only consider such 
%points $(x,y)$ where the endpoints of the ladder belong to segments $s_1,s_2$ 
%respectively. 
Since $(x-2\frac{x_1}{y_1}y)^2+(y)^2=r^2-(\frac{uy_1}{2})^2+(\frac{uy_1}{2})^2=r^2$, 
then we have $x^2+y^2(1+4(\frac{x_1}{y_1})^2)-4\frac{x_1}{y_1}xy=r^2$, so all points 
$(x,y)$ for a given $r$ lie on an ellipse - see Figure \ref{fig:krzywa}.

\begin{figure}[htbp]
\centering
\includegraphics[scale=0.4]{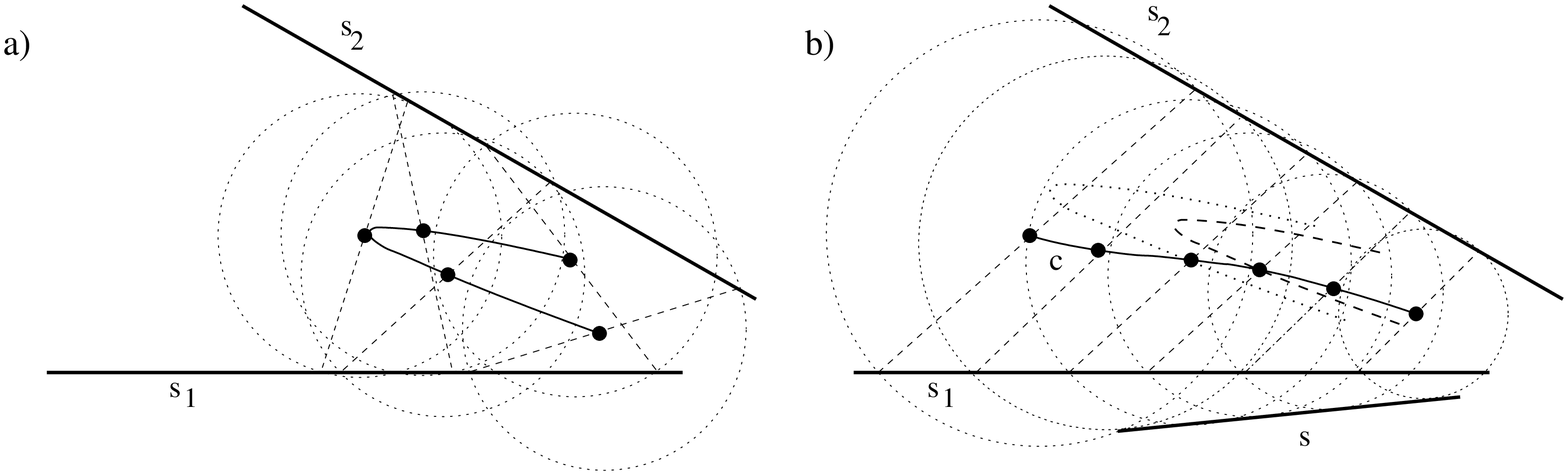}
\caption{(a) The set of middle points of segments $v_1v_2$, 
where $v_1 \in s_1$ and $v_2 \in s_2$} and (b) the curve $c$ such that the distance between a point on the curve $p$ and the segment $s$ is equal to the length of the radius of a corresponding disc centered in $p$.
\label{fig:krzywa}
\end{figure}

We want to find a point $p \in 3-VR(s_1,s_2,s)$ which is a center of a segment 
$v_1v_2$, where $v_1 \in s_1$ and $v_2 \in s_2$, and $d(p,s) > \frac{d(v_1,v_2)}{2}$.
 Then the disc with the center in $p$ and the radius $\frac{d(v_1,v_2)}{2}$
intersects only segments $s_1, s_2$, i.e., there exists an edge of $\mathit{GG}(S)$
between $s_1$ and $s_2$. 

We need to examine two cases. 
%First, let us consider the intersection of sets $T(s_1,s_2,s,r)$ for different 
%parameters $r$ and a stripe $M(s)$ defined by two lines, passing though segment $s$ 
%endpoints, perpendicular to this segment. 
First, we consider the situation when the closest to $p$ point of a segment $s$ belongs
to the interior of $s$.  
Let $P(s)$ be the line that contains segment $s$, which endpoints are $(x_1^s,y_1^s)$ 
and $(x_2^s,y_2^s)$, and let $q(t_s)=(x_1^s,y_1^s)+t_s \cdot [x_2^s-x_1^s,y_2^s-y_1^s]$ be 
the parametrization of $P(s)$. Let $L(s,r)$ be a line parallel to $P(s)$ with 
parametrization $l(t_L)=(x_1^L,y_1^L)+t_L\cdot [x_2^s-x_1^s,y_2^s-y_1^s]$ such that 
the distance between $P(s)$ and $L(s,r)$ is equal to $r$. We compute the intersection 
of the ellipse $x^2+y^2(1+4(\frac{x_1}{y_1})^2)-4\frac{x_1}{y_1}xy=r^2$ and the line $L(s,r)$. 
The result is 
$[x_1^L+t_L(x_2^s-x_1^s)]^2+[y_1^L+t_L(y_2^s-y_1^s)]^2-4\frac{x_1}{y_1}
[x_1^L+t_L(x_2^s-x_1^s)][y_1^L+t_L(y_2^s-y_1^s)]=r^2$, so $t_L$ satisfies an equation 
$At_L^2+Bt_L+C=r^2$ where coefficients $A,B,C$ are fixed and depend on 
$x_1,y_1,x_i^s,x_i^L,y_i^s,y_i^L$ for $i=1,2$.
This equation defines a curve $c$ (see Figure \ref{fig:krzywa}) which intersects 
corresponding ellipses. A point $p$ which belongs to a part of the ellipse that lies 
on the opposite side of the curve $c$ than the segment $s$ is a center of  a disc
which has a diameter $v_1v_2$, where $v_1 \in S_1$, $v_2 \in S_2$, and does not
intersect segment $s$. 

%For a given parameter $l$ if point $p_l \in T(s_1,s_2,s,r) \cap M(s)$ is the farthest 
%from segment $s$ then $p_l$ lies on the line parallel to $s$, tangent to set 
%$T(s_1,s_2,s,r)$.

%For intersection $T(s_1,s_2,s,r) \cap (\mathbb{R}^2 \setminus M(s))$ 
%For each set $T(s_1,s_2,s,r) \cap (\mathbb{R}^2 \setminus M(s))$ we find 
%the points of intersection of this set and discs $D_1(r)$ and $D_2(r)$. 
 In the second case one of the endpoints of the segment $s$ is the nearest point to $p$ (among the points from $s$). 
Let $D_1(r)$ and $D_2(r)$ be discs with diameter $r$ and with centers in corresponding 
ends of the segment $s$.  
We compute the intersection of $D_1(r)=\{(x,y):(x_1^s-x)^2+(y_1^s-y)^2=r^2\}$ 
($D_2(r)=\{(x,y):(x_2^s-x)^2+(y_2^s-y)^2=r^2\}$, respectively) and ellipse 
$x^2+y^2(1+4(\frac{x_1}{y_1})^2)-4\frac{x_1}{y_1}xy=r^2$. We obtain 
$x_1^s(x_1^s-2x)+y_1^s(y_1^s-2y)-y^2(\frac{x_1}{y_1})^2+4(\frac{x_1}{y_1})xy=0$, so 
$x=\frac{N_1y^2+N_2y+N_3}{N_4y+N_5}$ and $y$ satisfies an equation 
$M_1y^4+M_2y^3+M_3y^2+M_4y+M_5=0$ where coefficients $N_i$ and $M_j$ 
for $i,j=1, \dots, 5$ depend on $x_1^s,y_1^s,x_1,y_1,r$ (or on 
$x_2^s,y_2^s,x_1,y_1,r$, respectively). 
If there exists a point $p \notin D_1(r) \cup D_2(r)$ that belongs to 
the part of the ellipse between the segments $s_1,s_2$, then there also exists  
a disc with center in $p$ and a diameter $d(v_1,v_2)=2r$, where $v_1 \in s_1$ 
and $v_2 \in s_2$, which does not contain ends of the segment $s$.

In both cases we obtain a curve $c(r)$ dependent on the parameter $r$ 
- see Figure \ref{fig:krzywa}.
We check if a set $Q(s_1,s_2) \cap (2-VR(s_1,s_2)) \cap (3-VR(s_1,s_2,s))$
and the segment $s$ are on the same side of the curve $c$.
Otherwise, the segment $s_1s_2$ belongs to the Gabriel Graph for the set $S$
(i.e., there exists a point $p$ which is the center of a segment 
$v_1v_2$, where $v_1 \in s_1$, $v_2 \in s_2$, and $d(p,s) > \frac{d(v_1,v_2)}{2}$
for all $s \in S \setminus \{s_1,s_2\}$).

%We check if the points we found define empty Gabriel Graph lunes. 
%If so, then an edge exists between segments $s_1,s_2$ in $G_{1}^s(S)$. Else we check 
%a different segment $s \in S \setminus |\{s_1,s_2\}$.

\begin{theorem}
For a set of $n$ segments $S$ the Gabriel Graph $\mathit{GG}(S)$ can be computed in $O(n \log n)$ 
time.
\end{theorem}
\begin{proof}
The $2$-order Voronoi diagram and the $3$-order Voronoi diagram can be found in $O(n \log n)$ 
time \cite{pz13}. The number of triples of segments we need to test is linear. 
For each such triple we can check if there exists an empty $1$-skeleton lune
in time proportional to the complexity of the set 
$Q(s_1,s_2) \cap (2-VR(s_1,s_2)) \cap (3-VR(s_1,s_2,s))$.
The total complexity of those sets is $O(n)$. 
Hence, the complexity of the algorithm is $O(n)+O(n \log n) = O(n \log n)$.
\end{proof}

\section{Conclusions}

The running time of the presented algorithms for $\beta$-skeletons for sets 
of $n$ line segment ranges between $O(n\log n)$, $O(n^2 \alpha (n) \log n)$ 
and $O(n^3 \lambda_4(n))$ and depends on the value of $\beta$.  
For $0 < \beta < 1$ the $\beta$-skeleton 
is not related to the Delaunay triangulation of the underlying set of segments.
%However, a relatively efficient $O(n \log n)$ algorithm for the Gabriel Graph 
%($\beta = 1$), may suggest existence of more efficient algorithms for other values 
%of $\beta$, especially for $1 \le \beta \le 2$.
%Additional interesting questions about $\beta$-skeletons are related to their 
%connections with $k$-order Voronoi diagrams for line segments.
%The existence of this algorithm suggests that it may be possible to find a faster way 
%to compute $\beta$-skeletons for other values of $\beta$, especially 
%for $1 \le \beta \le 2$.\\
The existence of a relatively efficient algorithm for the Gabriel Graph 
suggests that it may be possible to find a faster way 
to compute $\beta$-skeletons for other values of $\beta$, especially 
for $1 \le \beta \le 2$.

 The edges of the Delaunay triangulation for line segments can be represented 
in the form described in this paper as rectangles contained in $[0,1] \times [0,1]$
square in the $t_1,t_2$-coordinate system.
If for each pair of $\beta$-skeleton edges the intersection of the corresponding
sets for the $\beta$-skeleton and the Delaunay triangulation is not empty then
there exist a plane partition generated by some pairs of generators of $\beta$-skeleton
neighborhoods. Unfortunately, it is not always possible (see Figure \ref{fig:diameter} in the Appendix).\\
The algorithms shown in this work for each pair of segments find such a position 
of generators that the corresponding lune does not intersect any other segment. 
We could consider a problem in which the number of used generators of neighborhoods is $n$ (one generator per each edge). Then the method described in the paper can also be used. 
We analyze a $n$-dimensional space and test if 
$[0,1]^n \setminus \bigcup_{s_i,s_j \in S, s\in S\setminus \{s_i,s_j\}} F(s_i,s,s_j) 
\times R^{n-2} \neq \emptyset$, where $i$ and $j$ also define corresponding coordinates 
in $R^n$. Unfortunately, such an algorithm is expensive. 
 However, in this case a $\beta$-skeleton already generates a plane partition.

 The total kinetic problem that can be solved in similar way is a construction 
$\beta$-skeletons for points moving rectilinear but without limitations concerning 
intersections of neighborhoods with lines defined by the moving points. In this case the form of sets 
$F(s_i,s,s_j)$ changes and the solution is much more complicated.

 Are there any more effective algorithms for those problems? 

Additional interesting questions about $\beta$-skeletons are related 
to their connections with $k$-order Voronoi diagrams for line segments.  

%We have described algorithms for computing $\beta$-skeletons 
%for a set of $n$ line segments $S$. For $0 <\beta < 1$, the running time is $O(n^4)$. 
%For $1 \leq \beta$ we can compute $G_{\beta}(S)$ in $O(n^2\alpha(n) \log n)$ time. 
%In the special case of Gabriel Graph we can solve this problem in $O(n \log n)$ time.
%An open question is if the $\beta$-skeleton can be computed faster (especially 
%for $1 \leq \beta \leq 2$). 
%Connections between $\beta$-skeletons and Voronoi diagrams for line segments
%(especially $k$-order Voronoi diagrams) are also unexplored.
%Many interesting questions concern properties of $\beta$-skeletons for segments,
%as well as practical applications, for example correlations between long magnetized 
%objects. 

\noindent
{\bf Acknowledgments}\\
\noindent
The authors would like to thank Jerzy W. Jaromczyk for important comments.

%%
%% Bibliography
%%

%% Either use bibtex (recommended), but commented out in this sample

%\bibliography{dummybib}

%% .. or use bibitems explicitely

%\nocite{Simpson}

\clearpage

\section{Appendix}

\subsection{Coefficients used in Section $3$:} 

$M = (y_2 - y_1)(A_1 \cos \delta + A_2 \sin \delta) - (x_2 - x_1)(-A_1 \sin \delta 
+ A_2 \cos \delta)$, 
$N = (y^2_2 - y^2_1)(A_1 \cos \delta + A_2 \sin \delta)$. \\

\subsection{Polynomials computed in Section $3$:} 

$p_1(t_1) = [(B_1 \cos \delta + B_2 \sin \delta)(y_2 - y_1) + (x^2_2 - x^2_1)
(-A_1 \sin \delta + A_2 \cos \delta) - (x_2 - x_1)(-B_1 \sin \delta + B_2 \cos \delta)]t_1 
+(A_1 \cos \delta + A_2 \sin \delta)(y_1 - y^2_1) - (x_2 - x_1)(-C_2 \cos \delta 
+ C_2 \cos \delta) + (x^2_1 - x_1)(-A_1 \sin \delta + A_2 \cos \delta)$, \\
$p_2(t_1) = (x^2_2 -x^2_1)(-B_1 \sin \delta + B_2 \cos \delta)t^2_1 
+ [(B_1 \cos \delta +B_2 \sin \delta)(y_1 - y^2_1) + (x^2_1 - x_1)(-B_1 \sin \delta 
+ B_2 \cos \delta) + (x^2_2 - x^2_1)(-C_1 \sin \delta +C_2 \cos \delta)]t_1 
+ (C_1 \cos \delta + C_2 \sin \delta)(y_1 - y^2_1) + (x^2_1 - x_1)(C_1 \sin \delta 
+ C_2 \cos \delta)$, \\
$p_3(t_1) = (y^2_2 -y^2_1)(B_1 \cos \delta + B_2 \sin \delta)t_1 + (y^2_2 - y^2_1)
(C_1 \cos \delta + C_2 \sin \delta)$.

\begin{figure}[htbp]
\centering
\includegraphics[scale=0.3]{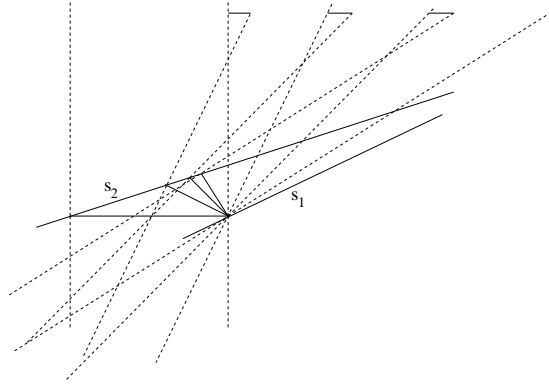}
\caption{An example of a set $S$ (for $\beta$ close to infinity) where 
the same point in $s_1$ generates lunes which correspond to points in different connected components of $[0,1] \times [0,1] \setminus \bigcup_{s \in S \setminus \{s_1,s_2\}} F(s_1,s,s_2)$.}
\label{fig:paski}
\end{figure}

\begin{figure}[htbp]
\centering
\includegraphics[scale=0.25]{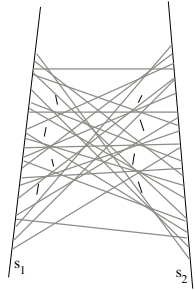}
\caption{
An example of a set $S$ where for segments $s_1, s_2$ the difference 
$[0,1] \times [0,1] \setminus \bigcup_{s \in S \setminus \{s_1,s_2\}} F(s_1,s,s_2)$
contains $\Omega(n^2)$ connected components (for a very small $\beta$). 
}
\label{fig:omega4}
\end{figure}

\begin{figure}[htbp]
\centering
\includegraphics[scale=0.55]{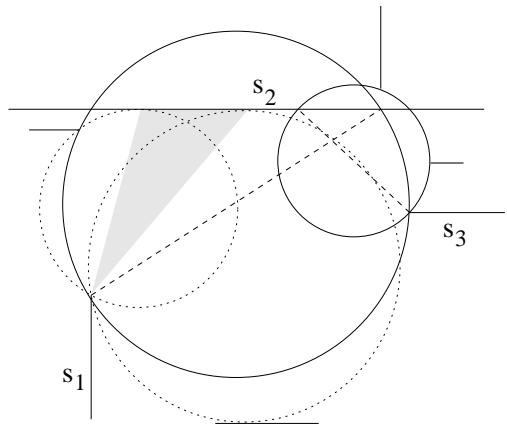}
\caption{
An example of a set $S$ where for segments $s_1, s_2$ there is no such Gabriel Graph neighborhood that the diameter of the disc defining this neighborhood is contained in the Delaunay Triangulation edge $s_1s_2$ (the light grey area) even though the edge $s_1s_2$ belongs to the Gabriel Graph for $S$. Moreover, generators of the neighborhoods of the {\em GG}-edges $s_1s_2$ and $s_2s_3$ create intersecting diameters between corresponding line segments.
}
\label{fig:diameter}
\end{figure}

\clearpage

\end{document}